\documentclass[preprint,12pt]{elsarticle}

\usepackage{amssymb}
\usepackage{amsmath}
\usepackage{amsthm}
\usepackage{caption}
\usepackage{subcaption}

\newtheorem{thm}{Theorem}[section]

\newtheorem{lemma}{Lemma}[section]

\newtheorem{obs}{Observation}[section]
\newtheorem{claim}{Claim}[section]

\usepackage{graphicx}
  \baselineskip 8 mm

\journal{}

\begin{document}
\begin{frontmatter}


\title{Algorithmic Results for Weak Roman Domination Problem in Graphs}

\author[inst1]{Kaustav Paul\corref{mycorrespondingauthor}}
\cortext[mycorrespondingauthor]{Corresponding author}
\ead{kaustav.20maz0010@iitrpr.ac.in}

\author[inst1]{Ankit Sharma}
\ead{ankit.23maz0013@iitrpr.ac.in}

\author[inst1]{Arti Pandey\footnote{Research of Arti Pandey is supported by CRG project, Grant Number-CRG/2022/008333, Science
and Engineering Research Board (SERB), India}}
\ead{arti@iitrpr.ac.in}

\affiliation[inst1]{
            addressline={Department of Mathematics, Indian Institute of Technology Ropar}, 
            city={Rupnagar},
            postcode={140001}, 
            state={Punjab},
            country={India}}



\begin{abstract}
Consider a graph $G = (V, E)$ and a function $f: V \rightarrow \{0, 1, 2\}$. A vertex $u$ with $f(u)=0$ is defined as \emph{undefended} by $f$ if it lacks adjacency to any vertex with a positive $f$-value. The function $f$ is said to be a \emph{Weak Roman Dominating function} (WRD function) if, for every vertex $u$ with $f(u) = 0$, there exists a neighbour $v$ of $u$ with $f(v) > 0$ and a new function $f': V \rightarrow \{0, 1, 2\}$  defined in the following way: $f'(u) = 1$, $f'(v) = f(v) - 1$, and $f'(w) = f(w)$, for all vertices $w$ in $V\setminus\{u,v\}$; so that no vertices are undefended by $f'$. The total weight of $f$ is equal to $\sum_{v\in V} f(v)$, and is denoted as $w(f)$. The \emph{Weak Roman Domination Number} denoted by $\gamma_r(G)$, represents $min\{w(f)~\vert~f$ is a WRD function of $G\}$. For a given graph $G$, the problem of finding a WRD function of weight $\gamma_r(G)$ is defined as the \emph{Minimum Weak Roman domination problem}. The problem is already known to be NP-hard for bipartite and chordal graphs. In this paper, we further study the algorithmic complexity of the problem. 
We prove the NP-hardness of the problem for star convex bipartite graphs and comb convex bipartite graphs, which are subclasses of bipartite graphs. In addition, we show that for the bounded degree star convex bipartite graphs, the problem is efficiently solvable. We also prove the NP-hardness of the problem for split graphs, a subclass of chordal graphs. On the positive side, we give polynomial-time algorithms to solve the problem for $P_4$-sparse graphs. Further, We have presented some approximation results.

\end{abstract}




\end{frontmatter}

\section{Introduction}
\label{sec:1}
 Cockayne et al.~\cite{COCKAYNE200411}, introduced the concept of \emph{Roman Dominating function} (RDF) for a graph $G$. Given a graph $G=(V,E)$, RDF is a function $f : V \rightarrow \{0, 1, 2\}$ that satisfies the following condition: for every vertex $u$ with $f(u) = 0$, there exists a vertex $v$ adjacent to $u$ with $f(v) = 2$. 
 The weight of $f$  is defined as the summation of the function values over all elements in $V$, and is denoted as $w(f)$. Furthermore, for a subset $S$ of $V$, we define $f(S)$ as the summation of the function values over all elements in $S$. 
 The \emph{Roman Domination number}, represented as $\gamma_R(G)$, is defined as $min\{w(f)~\vert~f$ is a RDF of
$G\}$. A function that represents an RDF of weight $\gamma_R(G)$ is referred as a $\gamma_R(G)$-function. Extensive research has been conducted on the topic of Roman domination on graphs, as evidenced by the work referenced in \cite{weak_rom_1,DBLP:journals/dmgt/Henning02,COCKAYNE200411}.

The impetus for formulating this description of a Roman dominating function stemmed from an article authored by Ian Stewart in Scientific American, named ``Defend the Roman Empire!" \cite{stewart1999defend}. Every vertex inside the graph corresponds to a distinct geographical region within the historical context of the Roman Empire. A place, denoted as node $v$, is classified as \emph{unsecured} if there are no legions stationed there, that is $f(v) = 0$. Conversely, a place is deemed \emph{secured} if there are one or two legions stationed there, represented by $f(v)\in \{1, 2\}$. One possible method of securing an insecure area, (or a vertex $v$), is by deploying a legion from a neighbouring location. In the fourth century A.D., Emperor Constantine the Great issued an edict stipulating that the transfer of a legion from a fortified position to an unfortified one is prohibited if such an action would result in leaving the latter unsecured. Therefore, it is necessary to have two legions positioned at a specific place ($f(v) = 2$) in order to afterward deploy one of the legions to a neighbouring place. Through this approach, Emperor Constantine the Great was able to provide a robust defense for the Roman Empire. Due to the considerable costs associated with maintaining a legion in a specific area, the Emperor aims to strategically deploy the minimum number of legions necessary to safeguard the Roman Empire.

The notion of \emph{Weak Roman domination} was initially introduced by Henning et al. \cite{weak_rom_1} to reduce the expense associated with the placement of legions. This problem involves easing some limitations that are imposed in the Roman Domination problem. Consider a graph $G=(V, E)$ and a function $f$ defined as $f : V \rightarrow \{0, 1, 2\}$. Let $V_0$, $V_1$, and $V_2$ denote the sets of vertices that are allocated the values $0$, $1$, and $2$, respectively, under the function $f$. A vertex $u \in V_0$ is said to be \emph{undefended} by $f$ (or simply \emph{undefended} if the function $f$ is evident from the context), if it is not adjacent to any vertex in $V_1\cup V_2$.

The function $f$ is referred to as a \emph{Weak Roman dominating function} (\emph{WRD function}) if, for every vertex $u$ in the set $V_0$, there exists a neighbour $v$ of $u$ in the set $V_1 \cup V_2$ and a function $f': V \rightarrow \{0, 1, 2\}$ defined in the following way: $f'(u) = 1$, $f'(v) = f(v) - 1$, and $f'(w) = f(w)$ for all vertices $w$ in the set $V\setminus\{u,v\}$, such that no vertices are undefended by $f'$. We define the weight $w(f)$ of $f$ to be $\sum_{v\in V}f(v)$. The \emph{Weak Roman domination number}, denoted as $\gamma_r(G)$, is the minimum weight of a WRD function in $G$; that is, $\gamma_r(G) = min\{w(f) |f$ is a WRD function in $G\}$. A WRD function of weight $\gamma_r(G)$ is denoted as a $\gamma_r(G)$-function. 
\begin{figure}[h!]
    \centering
    \includegraphics[width=0.98\linewidth]{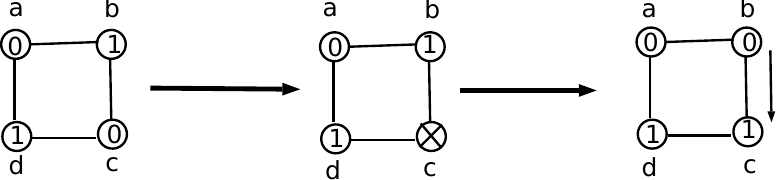}
    \caption{Example of a WRD function on $C_4$}
\label{fig:1}
\end{figure}

For example refer to Figure \ref{fig:1}, a graph $G=C_4$ is considered, where $V(G)=\{a,b,c,d\}$ and $E(G)=\{ab,bc,cd,da\}$. Consider the function $f:V(G)\rightarrow\{0,1,2\}$, defined as follows : $f(a)=f(c)=0$ and $f(b)=f(d)=1$. There is no undefended vertex by $f$. Now let an unsecured vertex be attacked (without loss of generality, say $c$; the ``crossed" vertex in figure \ref{fig:1}). The legion placed at $b$ can move to $c$ and the new function $f'$ ($f'(d)=f'(c)=1$ and $f'(a)=f'(b)=0$) is formed; but still $G$ has no vertex that is undefended by $f'$. Hence $f$ is a WRD function of $G$.

The optimization and decision version of the problem are the following:\\

\noindent\underline{\textsc{Minimum Weak Roman domination} problem (MIN-WRD)}
\begin{enumerate}
  \item[] \textbf{Instance}: A graph $G=(V,E)$.
  \item[] \textbf{Solution}: A WRD function $f$ of $G$ with $w(f)=\gamma_{r}(G)$.
\end{enumerate}

\noindent\underline{\textsc{Weak Roman domination decision} problem (DECIDE-WRD)}
\\
[-13pt]
\begin{enumerate}
  \item[] \textbf{Instance}: A graph $G=(V,E)$ and a positive integer $k\leq \vert V\vert$.
  \item[] \textbf{Question}: Does there exists a WRD function $f$ of $G$ with weight $w(f)\leq k$?
\end{enumerate}

\subsection{Notations and definitions}\label{subsec1.3}
This paper only considers simple, undirected, finite and nontrivial graphs. Let $G=(V,E)$ be a graph. $n$ and $m$ will be used to denote the cardinalities of $V$ and $E$, respectively. $N(v)$ stands for the set of neighbors of a vertex $v$ in $V$. The number of neighbors of a vertex $v\in V$ defines its \emph{degree}, which is represented by the symbol $deg(v)$. The maximum degree of the graph will be denoted by $\Delta$. For a set $U\subseteq V$, the notation $deg_{U}(v)$ is used to represent the number of neighbors that a vertex $v$ has within the subset $U$. Additionally, we use $N_{U}(v)$ to refer to the set of neighbors of vertex $v$ within $U$.  A vertex $u$ is called a \emph{private neighbor} of $v$ with respect to $S\subseteq V$ of $v$, if $N[u] \cap S = \{v\}$.
The set $pn(v; S)=N[v]-N[S\setminus {v}]$ of all private neighbours of $v$ with respect to $S$ is called the \emph{private neighbor set} of $v$ with respect to $S$.

A vertex of degree one is known as a \emph{pendant vertex}. A set $I\subseteq V$ is called an \emph{independent set} if no two vertices of $I$ are adjacent. A set $S\subseteq V$ is said to be a \emph{dominating set} if every vertex of $V\setminus S$ is adjacent to some vertex of $S$. A graph 
$G$ is said to be a \emph{complete graph} if any two vertices of $G$ are adjacent. A set $S\subseteq V$ is said to be a clique if the subgraph of $G$ induced on $S$ is a complete graph.$[n]$ denotes the set $\{1,2,\ldots,n\}$ for every non negative integer $n$. 

The \emph{join} of two graphs $G_{1}$ and $G_{2}$ refers to a graph formed by taking separate copies of $G_{1}$ and $G_{2}$ and connecting every vertex in $V(G_{1})$ to each vertex in $V(G_{2})$ using edges. The symbol $\oplus$ will denote the join operation. Similarly, \emph{disjoint union} of two graphs $H_1$ and $H_2$ is the graph $H=(V(H_1)\cup V(H_2), E(H_1)\cup E(H_2))$. The disjoint union will be denoted with the symbol $\cup$.

A vertex $v$ of a graph $G=(V,E)$ is said to be a \emph{universal vertex} if $N[v]=V$. A path with $n$ vertices is denoted as $P_n$. A \emph{comb} 
is constructed by adding a pendant vertex to each vertex of a path $P_n$. 

A graph $G=(V,E)$ is said to be $P_4$-sparse if a subgraph induced on any $5$ vertices of $G$ contains at most one $P_4$. A \emph{spider} is a graph $G=(V,E)$, where $V$ admits a partition in three subsets $S,C$ and $R$ such that
\begin{itemize}
     \item $C=\{c_1,\ldots,c_l\}~(l\geq 2)$ is a clique.
    \item $S=\{s_1,\ldots,s_l\}$ is a stable set.
    \item Every vertex in $R$ is adjacent to every vertex in $C$ and nonadjacent to all vertex of $S$.
\end{itemize}
A spider $G(S,C,R)$ is said to be a \emph{thin spider} if for every $i\in [l]$, $N(s_i)=\{c_i\}$ and it is called a \emph{thick spider} if for every $i\in [l]$, $N(s_i)=C\setminus\{c_i\}$. 

A graph $G=(V,E)$ is said to be \emph{bipartite} if $V(G)$ can be partitioned into two independent sets $X$ and $Y$. A bipartite graph $G=(X\cup Y,E)$ is called a \emph{complete bipartite graph} or a \emph{bi-clique} if every vertex of $X$ is adjacent to every vertex of $Y$. A complete bipartite graph $G=(X\cup Y,E)$ is called \emph{star} if any of $X$ or $Y$ has cardinality one.

A bipartite graph $G=(X\cup Y,E)$ is \emph{comb convex} (or \emph{star convex}) if a tree $T=(X,F)$ can be defined, such that $T$ is a comb (or a star) and for every vertex $y\in Y$, the neighbourhood of $y$ induces a subtree in $T$.

A graph is \emph{split} if it can be partitioned in an independent set and a clique.

\subsection{Existing Literature}\label{subsec1.1}
Henning et al. introduced the idea of weak roman domination in \cite{weak_rom_1}. In the same paper, they have provided the following bound for $\gamma_r(G)$ in terms of $\gamma(G)$: $\gamma(G)\leq\gamma_r(G)\leq 2\gamma(G)$ and characterization of graphs is given in which lower bound is attained. Additionally, it presents characterizations of forests in which the upper bound $\gamma_r(G)=2\gamma(G)$ is met. It has also been proved that the problem the DECIDE-WRD is NP-complete, even when restricted to bipartite and chordal graphs. On the positive side, they have shown that the MIN-WRD can be solved in linear time for trees. 

Liu et al. further studied the MIN-WRD problem, and proved that the problem can be solved in polynomial-time for block graphs \cite{liu2010weak}. Chapelle et al. proved that the MIN-WRD problem can be solved for general graphs with a time complexity of $O^*(2^n)$, accompanied by an exponential space complexity. The authors have also presented an algorithm with a time complexity of $O^*(2.2279^n)$, which takes polynomial space. In addition they proposed a polynomial algorithm to solve the MIN-WRD problem in interval graphs \cite{weak_rom_2}. To the best of our current understanding, no further algorithmic advancements are made regarding the weak roman domination problem.

\subsection{Our Results}\label{subsec1.2}
The subsequent sections of this work are organized in the following manner: 

\begin{itemize}
    \item In Section \ref{SEC:2}, a polynomial-time algorithm has been proposed for solving the MIN-WRD problem for $P_4$-sparse graphs.
    
    \item In Section \ref{SEC:3}, the DECIDE-WRD problem is shown to be NP-complete for comb convex bipartite graphs.

    \item In Section \ref{SEC:4}, the DECIDE-WRD problem is shown to be NP-complete for star convex bipartite graphs. In \cite{weak_rom_1}, it has been shown that the DECIDE-WRD problem is NP-complete for bipartite graphs, but using the same reduction, it can also be observed that the same problem is also NP-complete for bounded degree bipartite graphs.  We have shown that the MIN-WRD problem is polynomial-time solvable for bounded degree star convex bipartite graphs.

    \item In Section \ref{SEC:5}, the DECIDE-WRD problem is shown to be NP-complete for split graphs.

    
    \item In Section \ref{SEC:7}, a $2(1+ln(\Delta+1))$ approximation algorithm is given for the MIN-WRD problem. In the same section we have shown the APX-completeness of the problem for graphs with maximum degree $4$. Section \ref{SEC:8} concludes our work.

\end{itemize}

\section{Algorithm for $P_4$-sparse graphs}
\label{SEC:2}
In this section, we give a polynomial-time algorithm to solve MIN-WRD for $P_4$-sparse graphs. The class of $P_4$-sparse graphs can be considered as a generalization of the class of cographs. Below we present a characterization theorem for $P_4$-sparse graphs.

\begin{thm}\cite{Defn_P4sparse}\label{th:5}
    A graph $G$ is said to be $P_4$-sparse if and only if one of the following conditions hold
    \begin{itemize}
        \item $G$ is a single vertex graph.
        \item $G=G_1\cup G_2$, where $G_1$ and $G_2$ are $P_4$-sparse graphs.
        \item $G=G_1 \oplus G_2$, where $G_1$ and $G_2$ are $P_4$-sparse graphs.
        \item $G$ is a spider which admits a spider partition $(S, C, R)$ where either $G[R]$ is a $P_4$-sparse graph or $R=\phi$.
    \end{itemize} 
\end{thm}

Hence by Theorem \ref{th:5}, a connected graph that is $P_4$-sparse and contains at least two vertices can be classified as either a join of two $P_4$-sparse graphs or a particular type of spider (thick or thin). In this section, we discuss the Weak Roman domination number for both of these cases. 

\begin{lemma}\label{lem:1}
    Let $G=(V,E)$ be a headless thin spider ($R=\phi$ in the spider partition $(S,C,R)$), then, $\gamma_{r} (G)$ is equal to $\vert C\vert$.
\end{lemma}

\begin{proof}
Consider a graph $G=(V,E)$ that is a thin spider. Let the spider partition of $G$ be represented as $(S,C,R)$, where $R$ is an empty set. In \cite{weak_rom_1}, it is proved that $\gamma(G)\leq \gamma_r(G)$. Note that for the graph $G$, $\gamma(G)=\vert C\vert$. Now it is enough to show the existence of a WRD function with weight equal to $\vert C\vert$ in order to establish the lemma. We define a function $f:V\rightarrow \{0,1,2\}$, as follows: for any element $u$ in the set $S$, $f(u)$ equals $1$, and for any other element $u$, $f(u)$ equals $0$. Clearly, $f$ is a WRD function and $w(f)=\vert C\vert$.
\end{proof}

\begin{lemma}\label{lem:2}
    Let $G=(V,E)$ be a thin spider with a nonempty head, ($R\neq\phi$ in the spider partition $(S,C,R)$). Then, $\gamma_{r} (G)$ is equal to $\vert C\vert+1$.
\end{lemma}

\begin{proof}
    Consider a graph $G=(V,E)$ that is a thin spider. Let the spider partition of $G$ be represented as $(S,C,R)$, where $R$ is a nonempty set, and let $\vert C\vert=\vert S\vert=k$. Note that for the graph $G$, $\gamma(G)=\vert C\vert$. Now to complete the proof of this lemma we show the following two statements:
    \begin{itemize}
        \item[i.] There exists a WRD function $f$ on $G$ with $w(f)=\vert C\vert +1$.
        \item[ii.] $\gamma(G)<\gamma_r(G)$.
    \end{itemize}
    
    Initially, we establish a WRD function with a weight of $\vert C\vert +1$. Consider the function $f:V\rightarrow \{0,1,2\}$, defined as follows: $f(c_i)=2$ for a fixed $i\in [k]$ and $f(c_j)=1$ for all $c_j$ in the set $C\setminus \{c_i\}$. For any other element $u$, $f(u)$ is equal to $0$. It can be observed that the function $f$ exhibits the property of being a WRD function with a weight of $\vert C\vert+1$.

    Now we show that $\gamma(G)<\gamma_{r}(G)$. For sake of contradiction, assume that $\gamma(G)=\gamma_{r}(G)=\vert C\vert$. Then there exists a WRD function of $G$ with $w(f)=\gamma(G)=\vert C\vert$. In this case we can say that $f(c_j)+f(s_j)=1$ holds for all $j\in [k]$. This statement suggests that $f(r)$ equals zero for every $r$ in the set $R$. Therefore, an attack on a certain vertex $r$ belonging to the set $R$ can be defended by sending a legion from a certain element $c_j$ from the set $C$, where $f(c_j)$ equals $1$ (that means one legion is placed at $c_j$). However, following the defense, the function is transformed into $f'$, where $f'(r)=1$, $f'(c_j)=0$, and $f(u)=f'(u)$ for all $u\in V\setminus\{c_j,r\}$. This results in the vertex $s_j\in S$ being left undefended, which contradicts the initial assumption. Therefore, the value of $\gamma_r(G)$ is equal to $\vert C\vert+1$. 
\end{proof}

\begin{lemma}\label{lem:3}
    Let $G=(V,E)$ be a thick spider with an empty head, ($R=\phi$ in the spider partition $(S,C,R)$). Then, $\gamma_{r} (G)=2$ if $\vert C\vert=2$ and  $\gamma_{r} (G)=3$, for $\vert C\vert \geq 3$.
\end{lemma}

\begin{proof}
    Consider a graph $G=(V,E)$ that is a thick spider, such that $R=\phi$. Let $\vert C\vert=\vert S\vert=k$. If $k=2$, then $G$ is a $P_4$, hence $\gamma_r(G)=2$. For each integer $k$ greater than or equal to $3$, we show the existence of a WRD function $f$ such that the weight of $f$ is equal to $3$. The function $f:V\rightarrow \{0,1,2\}$ is defined as follows : $f(c_i)=2$ (for a fixed $c_i\in C$), $f(s_i)=1$ (for $s_i\in S$) and $f(u)=0$ for every $u\in V\setminus\{c_i,s_i\}$. It can be observed that the function $f$ exhibits the properties of being a WRD function and $w(f)=3$.

    Next to show that $\gamma_r(G)=3$ for $\vert C\vert\geq 3$, we show that $\gamma_r(G)$ can not be $2$. Assuming $\gamma_r(G)=2$ and $\vert C\vert\geq 3$, the following cases can be considered.
    
    \noindent\textbf{Case 1:} In this case we assume that $f(c_i)=f(c_j)=1$ for some fixed $c_i,c_j\in C$, and $f(u)=0$ for every $u\in V\setminus\{c_i,c_j\}$. If the attack is on the vertex $s_b$ (where $s_b$ is not a member of the set $\{s_i,s_j\}$), it can be assumed without loss of generality that the legion stationed at $c_i$ will be required to move to $s_b$, resulting in the vertex $s_j$ being undefended. Therefore, a contradiction emerges.

   \noindent\textbf{Case 2:} In the second case, the function $f$ is defined so that $f(c_i) = f(s_i) = 1$ for some $c_i$ in the set $C$ and corresponding $s_i$ in the set $S$. Additionally, $f(u) = 0$ for all elements $u$ in the set $V\setminus\{c_i,s_i\}$. If the attack is on $s_j$ (where $s_j$ is not equal to $s_i$), it necessitates the movement of the legion stationed at $c_i$ to $s_j$, resulting in the vertex $s_b$ being undefended ($s_b$ being an element of $S\setminus\{s_i,s_j\}$). This gives rise to a contradiction.

   \noindent\textbf{Case 3:} The last case being, there exists a WRD function $f$ defined as: $f(s_i)=f(s_j)=1$ for some fixed $s_i,s_j\in S$, and $f(u)=0$ for every $u\in V\setminus\{s_i,s_j\}$. Now there exists a vertex $s_b\in S$ (where $b\notin \{i,j\}$) as $\vert S\vert\geq 3$; hence $s_b$ is undefended by $f$, which contradicts the fact that $f$ is a WRD function.

   Consequently, it follows that $\gamma_r(G)$ is greater than or equal to $3$, which in turn implies that $\gamma_r(G)=3$.
 \end{proof}

 \begin{lemma}\label{lem:4}
    Let $G=(V,E)$ be a thick spider with a nonempty head, ($R\neq\phi$ in the spider partition $(S,C,R)$). Then, $\gamma_{r} (G)$ is equal to $3$.
\end{lemma}

\begin{proof}
    Consider a graph $G=(V,E)$ that is a thick spider, such that $R\neq\phi$. Let $\vert C\vert=\vert S\vert=k$. At first, we will show the existence of a WRD function $f$ such that the weight of $f$ is equal to $3$. Define the function $f:V\rightarrow \{0,1,2\}$ as follows: for some fixed $i\in [k]$, $f(c_i)=2$ and $f(s_i)=1$, where $c_i\in C$ and $s_i\in S$ and $f(u)=0$ for every $u\in V\setminus\{c_i,s_i\}$. It can be observed that the function $f$ exhibits the properties of being a WRD function and $w(f)=3$. If we assume that $\gamma_r(G)=2$, we can analyse three different possibilities. 
     
     \noindent\textbf{Case 1:} In the first scenario, it is seen that $f(c_i)=f(c_j)=1$ for two fixed $c_i$ and $c_j$ belonging to the set $C$, whereas $f(u)=0$ for all $u$ in the set $V$ excluding $c_i$ and $c_j$. In the event of an attack on a vertex $r\in R$, it is reasonable to suppose, without loss of generality, that the legion positioned at $c_i$ will need to relocate to $r$. Consequently, this relocation will leave the vertex $s_j$ undefended. Hence, a contradiction arises.

     \noindent\textbf{Case 2:} In the second scenario, the function $f$ is defined so that $f(c_i) = f(s_i) = 1$ for a fixed $c_i$ in the set $C$ and corresponding $s_i$ in the set $S$. Additionally, $f(u) = 0$ for all elements $u$ in the set $V\setminus\{c_i,s_i\}$. In the event that an attack is on $r\in R$, it becomes imperative to relocate the legion now stationed at $c_i$ to $r$. Consequently, this relocation leaves all the vertices of $S\setminus\{s_i\}$ undefended. This results in a contradiction.

     \noindent\textbf{Case 3:} The third scenario involves $f(r)=1$ and $f(u)=1$, for some fixed $r\in R$ and $u\in C\cup S$ and $f(v)=0$ for every $v\in V\setminus\{r,u\}$. Assume that there exists an element $c_i$ in the set $C$ such that $f(c_i)=1$. However, it follows that the element $s_i$ remains undefended. Likewise, in the case where $f(s_i)=1$ for a fixed $s_i\in S$, it follows that all vertices in $S$ except for $s_i$ remain undefended. Therefore, this statement presents a contradiction. Hence $\gamma_r(G)=3$.
\end{proof}

Now the only case that remains to be analyzed is the case when $G$ is a join of two $P_4$-sparse graphs $G_1$ and $G_2$. In this case, $\gamma(G)\leq 2$. Since $\gamma(G)\leq 2 \gamma_r(G)$ \cite{weak_rom_1}; $\gamma_r(G)\leq 4$. Now it is easy to observe that for a complete graph $G$, $\gamma_r(G)=1$ and vice versa. So from now on, we will consider that $G$ is not a complete graph. Hence, it is enough to characterize $G$ for $\gamma_r(G)=2,3$. Here, we are stating a theorem from \cite{weak_rom_1}, which will be useful in upcoming proofs.

\begin{thm}\cite{weak_rom_1}\label{th:6}
    For any graph $G$, $\gamma(G)=\gamma_r(G)$ if and only if there exists a minimum dominating set $S$ of $G$, such that
    \begin{enumerate}
        \item[(a)] $pn(v; S)$ induces a clique for every $v \in S$,
        \item[(b)] for every vertex $u \in V(G)\setminus S$ that is not a private neighbor of any vertex of $S$, there exists a vertex $v \in S$ such that $pn(v; S) \cup \{u\}$ induces a clique.
    \end{enumerate}
\end{thm}

Hence by Theorem \ref{th:6}, the following observation follows,
\begin{obs}\label{obs:2}
    Let $G=G_1\oplus G_2$, where $G_1$ and $G_2$ are $P_4$-sparse graphs. Then $\gamma_r(G)=2$ if and only if one of the following conditions satisfies
    \begin{enumerate}
        \item[(i)] $\gamma(G)=1$ and $G$ is not a complete graph.
        \item[(ii)] $\gamma(G)=2$ and there exists a minimum dominating set $S$ of $G$ that satisfies condition $(a)$ and $(b)$ of Theorem \ref{th:6}.
    \end{enumerate}
\end{obs}

Hence in the following lemma, we will only consider those graphs $G$, which are not a complete graph and do not satisfy the conditions $(i)$ and $(ii)$ of Observation \ref{obs:2}.

\begin{lemma}\label{lem:5}
    Let $G=G_1\oplus G_2$, where $G_1$ and $G_2$ are $P_4$-sparse graphs. Then $\gamma_r(G)=3$ if and only if one of the following conditions satisfies
    \begin{enumerate}
        \item[(i)] There exists one vertex in $u\in V(G)$ such that $G[V(G)\setminus N[u]]$ is a clique.
        \item[(ii)] There exists a set of two vertices $S=\{u_1,u_2\}\subseteq V(G^*)$ (where $G^*\in \{G_1,G_2\})$ such that $G^*[V(G^*)\setminus (N[u_1]\cup N[u_2])]$ is a clique.
    \end{enumerate}
\end{lemma}

\begin{proof}
    If $(i)$ holds, define a function $f:V(G)\rightarrow \{0,1,2\}$ in following way: $f(u)=2$ and $f(v)=1$ where $v\in V(G)\setminus N[u]$ and $f(x)=0$ for every $x\in V(G)\setminus\{u,v\}$. Clearly, $f$ is a WRD function, hence $\gamma_r(G)=w(f)=3$. If $(ii)$ holds, then let without loss of generality there exists a set of two vertices $S=\{u_1,u_2\}\subseteq V(G_1)$ such that $G_1[V(G_1)\setminus (N[u_1]\cup N[u_2])]$ is a clique. Now define a function $g:V(G)\rightarrow \{0,1,2\}$ in following way: $g(u_1)=g(u_2)=1$, $g(u_3)=1$ where $u_3\in V(G_2)$ and $g(x)=0$ for every $x\in V(G)\setminus\{u_1,u_2,u_3\}$. Clearly, $g$ is a WRD function, hence $\gamma_r(G)=w(g)=3$.
    
    Conversely, let $\gamma_r(G)=3$ and $f$ be a WRD function of weight $3$. Now two cases may appear. 
    
    \noindent \textbf{Case 1:} There exists two vertices $u_1$, $u_2$ in $V(G)$ such that $f(u_1)=2$ and $f(u_2)=1$. Let $S=\{u_1,u_2\}$. Now if $G[pn(u_2,S)]$ is not a clique then there exists two nonadjacent vertices $v,v'\in pn(u_2,S)$. In that case, any attack on $v$ must result in the legion placed at $u_2$ being moved to $v$, which leaves $v'$ undefended. Hence $G[pn(u_2,S)]$ is a clique, implying $G[V(G)\setminus N[u_1]]$ is a clique.
    
    \noindent \textbf{Case 2:} There exists three vertices $u_1$, $u_2$ and $u_3$ in $V(G)$ such that $f(u_1)=1$, $f(u_2)=1$ and $f(u_3)=1$. 
    Before analyzing this case we will prove the following claim.

    \begin{claim}\label{claim:1}
        Let $f$ be a minimum WRD function of $G=G_1\oplus G_2$ with $f(u_1)=f(u_2)=f(u_3)=1$, where $u_1,u_2,u_3\in V(G_1)$ and $f(v)=0$, otherwise; then there exists a minimum WRD function $f'$ on $G$ such that $f'(u_1)=f'(u_2)=1$ and $f'(u_3')=1$ for $u_3'\in V(G_2)$ and $f(v)=0$, otherwise.
    \end{claim}

    \begin{proof}
       We need to show that the function $f'$ effectively defends against all potential attacks. 
        
        Consider an arbitrary $v \in V(G_2)$ with $f(v)=0$. Suppose that this vertex is attacked. In response, the legion stationed at vertex $u_1$ moves to defend vertex $v$. It is easy to observe that after this move, all vertices in the graph remain defended.
        
        Let us consider an unsecured vertex $v$ (in the context of $f'$) belonging to the set $V(G_1)$, which is subjected to an attack. In the configuration $f$, it is necessary for one of the vertices $u_1, u_2, u_3$ to be moved at vertex $v$ in order to protect against the attack. Following this move, no unsecured vertices will be undefended (as $f$ is a WRD function). If a vertex from the set $\{u_1,u_2\}$ relocates to vertex $v$ in order to defend against the attack in the context of the function $f$, then the same vertex relocates to vertex $v$ in the case of the function $f'$. Similarly, if vertex $u_3$ relocates to vertex $v$ in the context of function $f$, then vertex $u_3'$ relocates to vertex $v$ in the case of function $f'$. 
    \end{proof}

    Hence by Claim \ref{claim:1}, it is enough to check the case when $u_1,u_2\in V(G_1)$ (or $V(G_2)$) and $u_3\in V(G_2)$ (or $V(G_1)$). Without loss of generality, let
    $u_1,u_2\in V(G_1)$ and $u_3\in V(G_2)$. Let $G_1[V(G_1)\setminus (N[u_1]\cup N[u_2])]$ is not a clique, then there exists two vertices $v,v'\in V(G_1)\setminus (N[u_1]\cup N[u_2])$ which are non adjacent. Then any attack on $v$ must be defended by the legion stationed at $u_3$ by moving to $v$; which leaves $v'$ undefended, which is a contradiction. Hence $G_1[V(G_1)\setminus (N[u_1]\cup N[u_2])]$ is a clique.
\end{proof}

Since the checks mentioned in lemma \ref{lem:5} and observation \ref{obs:2} can be done in time $O(n^3)$, the following theorem can be stated.

\begin{thm}\label{th:7}
    Given a connected $P_4$-sparse graph $G$; $\gamma_r(G)$ can be computed in $O(n^3)$ time. 
\end{thm}

\section{NP-hardness for Comb Convex Bipartite graphs}
\label{SEC:3}
In this section, we establish that the the DECIDE-WRD problem is NP-hard for comb convex bipartite graphs. To prove this result we present a polynomial-time reduction from the Restricted Exact Cover by $3$-Sets problem.

The Restricted Exact Cover by $3$-Sets (RXC$3$) problem is defined as follows: Given a finite set $X=\{x_1,\ldots,x_{3q}\}$, along with a collection $C=\{c_1,\ldots,c_{3q}\}$ (where $c_i \subseteq X$) comprising subsets of $X$, each containing exactly three elements from $X$, and ensuring that every element $x_i\in X$ is included in exactly three sets belonging to $C$. The objective of this problem is to check whether a subcollection, denoted as $C'(\subseteq C)$, such that each element $x_i\in X$ is contained in exactly one $c_i\in C'$, exists or not. This problem denoted as RXC$3$, was defined and proved to be NP-hard in \cite{GONZALEZ1985293}. We have shown a reduction from an instance of RXC$3$ problem to an instance of the DECIDE-WRD problem for comb convex bipartite graph. The reduction is explained in Construction $\mathcal{A}_1$.

\medskip
\noindent \textbf{Construction $\mathcal{A}_1$:} Consider an instance $(X, C)$ of the RXC$3$ problem such that $X=\{x_1,\ldots,x_{3q}\}$ and $C=\{c_1,\ldots,c_{3q}\}$. We proceed to construct a graph $H=(V',E')$ using the following procedure. First, we create three copies of the set $X$, denoted as $X_1, X_2,$ and $X_3$, as well as two copies of the set $C$, labeled as $C_1$ and $C_2$. Specifically, we define $X_i=\{x_{1}^{i},\ldots, x_{3q}^{i}\}$ for $i\in [3]$, $C_i=\{c_{1}^i,\ldots,c_{3q}^i\}$, for $i\in [2]$. Now $V'=X_1\cup X_2\cup X_3\cup C_1\cup C_2$ and $E'=\{c_i^{1}x_{j}^{1},c_i^{1}x_{j}^{2}~\vert~i,j\in [3q]\}\cup \{c_i^{1}x_{j}^{3}~\vert~x_j\in c_i\}\cup \{c_i^1c_i^2~\vert~i\in [3q]\}$.

\medskip
\noindent \textbf{Example 1:} $(X,C)$ be an instance of exact cover where $X=\{x_1,x_2\ldots,x_6\}$ and $C=\{c_1=\{x_1,x_4,x_5\},~c_2=\{x_2,x_4,x_3\},~c_3=\{x_2,x_3,x_5\},~c_4=\{x_1,x_4,x_6\},~c_5\\=\{x_1,x_2,x_6\},~c_6=\{x_3,x_5,x_6\}\}$. The corresponding $H$ from this specific instance has been given in Figure $\ref{fig:2}$.

\begin{figure}[h!]
    \centering
        \includegraphics[width=8 cm,height= 8.5 cm]{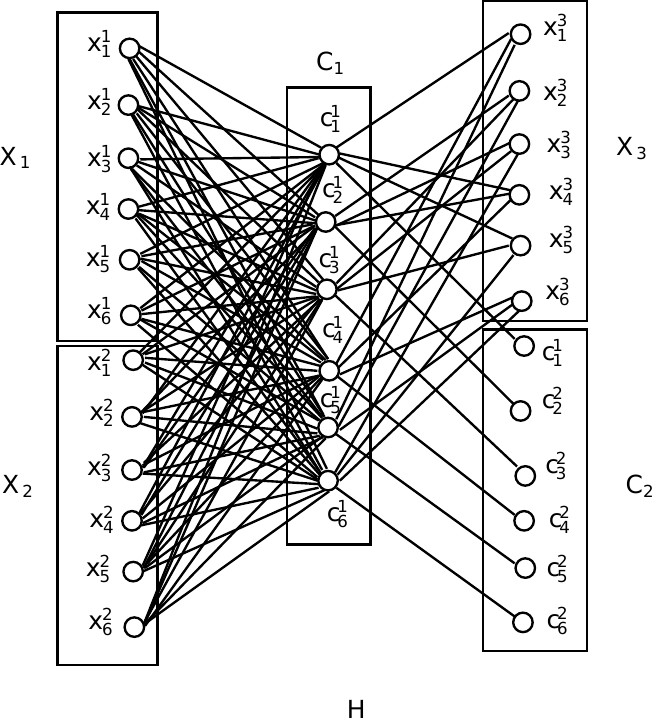}
    \caption{The Construction $\mathcal{A}_1$ applied on the Example 1}
\label{fig:2}
\end{figure}

Clearly $H$ is a bipartite graph,as $X_1\cup X_2\cup X_3\cup C_2$ and $C_1$ are two partitite sets. To establish that the graph $H$ is indeed a comb convex bipartite graph, we create a comb graph denoted by $F$ as follows: 

Let $V(F)=X_1\cup X_2\cup X_3\cup C_2$, with the backbone of the comb being $X_1\cup X_2$, and the vertices in $X_3\cup C_2$ acting as pendant vertices. It can be readily verified that for each $c_i^1\in C_1$, the neighborhood $N(c_i^1)$ in graph $F$ remains connected. Consequently, we conclude that $H$ is a comb convex bipartite graph.

\begin{thm}\label{th:1}
     $(X, C)$ is a ``yes" instance if and only if the graph $H$ possesses a WRD function with weight $4q$.  
\end{thm}

\begin{proof}
    Given an instance $(X, C)$ of the RXC$3$ problem, we can construct a comb convex bipartite graph $H$ using construction $\mathcal{A}_1$. It is important to note that this construction can be done in polynomial-time.
    
    Consider a ``yes" instance of the RXC$3$ problem $(X,C)$ such that $X=\{x_1,\ldots,x_{3q}\}$ and $C=\{c_1,\ldots,c_{3q}\}$. Let $C'\subseteq C$ be an exact cover of the instance $(X,C)$. It is important to note that for any such instance, the size of an exact cover (if one exists) is $\frac{\vert C\vert}{3}$. Hence $\vert C'\vert=q$. Now, we define a function $f:V(H)\rightarrow \{0,1,2\}$ as follows:
    \begin{enumerate}
        \item Define $C_1'=\{c_i^1~\vert~c_i\in C'\}$. For every vertex $u$ in $C_1'$, we set $f(u)=2$.
        \item For every vertex $u$ in $C_1\setminus C_1'$ (the corresponding sets in $C$, which are not included in the exact cover), we set $f(u)=1$.
        \item For all other vertices $u$ in $H$, $f(u)=0$.
    \end{enumerate}
    It is evident that this function $f$ constitutes a WRD function with a weight precisely equal to $4q$.

    Conversely, consider a WRD function $f$ with a weight of at most $4q$. It is important to note that for any vertex $c_i^1\in C_1$ and its corresponding counterpart $c_i^2\in C_2$ in the constructed graph $H$, we have the property $f(c_i^1) + f(c_i^2) \geq 1$. Therefore, the sum of the labels for all $c_i^1$ and $c_i^2$ is at least $3q$. As a consequence, we can deduce that $f(X_3)$ must be less than or equal to $q$.

    Now, let us assume that there exist $l$ vertices in $X_3$ with nonzero labels (where $0 < l \leq q$). This implies that $3q - l$ vertices in the set $X_3$ must have a label of $0$. Now, the following claim will help us to complete the proof.

    \begin{claim}\label{claim:2}
        For each $x\in X_3$ with $f(x)=0$; there exists a neighbour $c_i^1$ of $x$ with $f(c_i^1)+f(c_i^2)\geq 2$.
    \end{claim}

    \begin{proof}
        Let there exists $x\in X_3$ such that $x$ has no neighbours $c_i^1\in C_1$, with the property: $f(c_i^1)+f(c_i^2)\geq 2$ (which means $f(c_i^1)+f(c_i^2)= 1$ for every $c_i^1\in C_1$). Now consider an attack on $x$. A legion must come from one of its neighbours in $C_1$. Let from $c_j^1$ one legion comes to $x$, but these move leaves the vertex $c_j^2$ undefended (as $f(c_i^1)+f(c_i^2)= 1$), contradicting the fact that $f$ is a WRD function. Hence the claim holds.
    \end{proof}
    
    Since there are $3q - l$ vertices in the set $X_3$ with a label of $0$, and for each of these vertices, there exists a neighbor in $C_1$ (say $c_i^1)$ that satisfies $f(c_i^1) + f(c_i^2) \geq 2$, hence by this fact and Claim \ref{claim:2}, we can infer that there must be at least $\frac{3q - l}{3}$ vertices among the $c_i^1$s that satisfy $f(c_i^1) + f(c_i^2) \geq 2$. This implies that the weight of the WRD function $f$ is greater than or equal to $\frac{3q - l}{3} + 3q + l$ which simplifies to $4q + \frac{2l}{3}$. However, this conclusion contradicts the initial assumption that $f$ has a weight of at most $4q$. Hence it follows that $l$ must be equal to $0$. Therefore, $f(X_3) = 0$, and there exist exactly $\frac{3q}{3} = q$ vertices in the set $C_1$ that satisfy the condition $f(c_i^1) + f(c_i^2) = 2$. As a result, $C' = \{c \in C~\vert~f(c^1) + f(c^2) = 2\}$ forms an exact cover of the set $X$.
\end{proof}

Since checking the membership of the DECIDE-WRD problem in NP is easily verifiable, we can conclude by Theorem \ref{th:1}, that the DECIDE-WRD is NP-complete. 

\section{NP-hardness for Star Convex Bipartite graphs}
\label{SEC:4}
In this section, our objective is to establish the NP-completeness of the DECIDE-WRD problem for star convex bipartite graphs. To accomplish this, we will provide a polynomial-time reduction from the RXC$3$ problem to the DECIDE-WRD problem for star convex bipartite graphs. The construction is delineated as follows:

\noindent \textbf{Construction $\mathcal{A}_2$:}  Given an instance $(X,C)$ of the RXC$3$ problem, where $X=\{x_1,\ldots,x_{3q}\}$ and $C=\{c_1,\ldots,c_{3q}\}$; define $V'=\{a\}\cup C_1 \cup C_2 \cup X'$. Here, $X'=\{x_{1}',\ldots, x_{3q}'\}$, $C_1=\{c_{1}^1,\ldots, c_{3q}^1\}$, and $C_2=\{c_{1}^2,\ldots, c_{3q}^2\}$. The edge set $E'$ is defined as follows, $E'=\{x_j'c_i^1~\vert~x_j \in c_i\}\cup \{ac_i^1~\vert~c_i^1\in C_1\}\cup \{c_i^1c_i^2~\vert~i\in [3q]\}$.

Now a star $F$ can be made with vertices of $\{a\}\cup X'\cup C_2$ where $a$ is the central vertex of the star; such that $F[N(c_i^1)]$ is connected for every $i\in [3q]$. So, $H$ is a star convex bipartite graph. In summary, given an instance of the RXC$3$ problem, we can effectively construct a star convex bipartite graph $H$ using the defined methodology. Importantly, this construction can be executed in polynomial-time. Please refer to Figure \ref{fig:3} for a representation of the construction.

\begin{figure}[h!]
    \centering
    \includegraphics[width=5 cm,height= 7 cm]{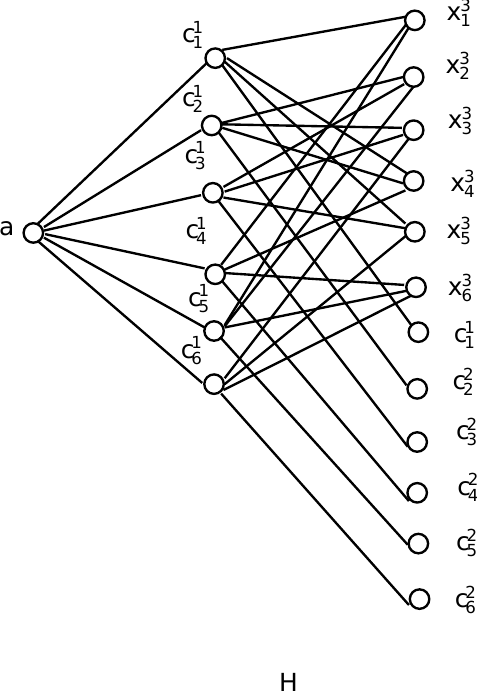}
    \caption{The Construction $\mathcal{A}_2$ applied on the Example 1}
    \label{fig:3}
\end{figure}

\begin{thm}\label{th:2}
    $(X,C)$ has an exact cover if and only if $H$ has a WRD function of weight $4q$. 
\end{thm}

\begin{proof}
    Given an instance $(X, C)$ of the RXC$3$ problem, we can construct a star convex bipartite graph $H$ using construction $\mathcal{A}_2$. It is important to note that this construction can be done in polynomial-time.
    
    Consider a ``yes" instance of the RXC$3$ problem $(X,C)$ such that $X=\{x_1,\ldots,x_{3q}\}$ and $C=\{c_1,\ldots,c_{3q}\}$. Let $C'\subseteq C$ be an exact cover of the instance $(X,C)$. It is important to note that for any such instance, the size of an exact cover (if one exists) is $\frac{\vert C\vert}{3}$. Hence $\vert C'\vert=q$. Now, we define a function $f:V(H)\rightarrow \{0,1,2\}$ as follows:
    \begin{enumerate}
        \item Define $C_1'=\{c_i^1~\vert~c_i\in C'\}$. For every vertex $u$ in $C_1'$, we set $f(u)=2$.
        \item For every vertex $u$ in $C_1\setminus C_1'$ (the corresponding sets in $C$, which are not included in the exact cover), we set $f(u)=1$.
        \item For all other vertices $u$ in $H$, $f(u)=0$.
    \end{enumerate}
    It is evident that this function $f$ constitutes a WRD function with a weight precisely equal to $4q$.

    Conversely, consider a WRD function $f$ with a weight of at most $4q$. It is important to note that for any vertex $c_i^1\in C_1$ and its corresponding counterpart $c_i^2\in C_2$ in the constructed graph $H$, we have the property $f(c_i^1) + f(c_i^2) \geq 1$. Therefore, the sum of the labels for all $c_i^1$ and $c_i^2$ is at least $3q$. As a consequence, we can deduce that $f(X_3)$ must be less than or equal to $q$.

    Now, let us assume that there exist $l$ vertices in $X'$ with nonzero labels (where $0 < l \leq q$). This implies that $3q - l$ vertices in the set $X'$ must have a label of $0$. Now, the following claim will help us to complete the proof.

    \begin{claim}\label{claim:3}
        Given a WRD function $f$ on $H$, for each $x\in X'$ with $f(x)=0$; there exists a neighbour $c_i^1$ of $x$ with $f(c_i^1)+f(c_i^2)\geq 2$.
    \end{claim}

    \begin{proof}
         Let there exists $x\in X'$ such that $x$ has no neighbours $c_i^1\in C_1$, with the property: $f(c_i^1)+f(c_i^2)\geq 2$ (which means $f(c_i^1)+f(c_i^2)= 1$ for every $c_i^1\in C_1$). Now consider an attack on $x$. A legion must come from one of its neighbours in $C_1$. Let from $c_j^1$ one legion comes to $x$, but these move leaves the vertex $c_j^2$ undefended (as $f(c_i^1)+f(c_i^2)= 1$), contradicting the fact that $f$ is a WRD function. Hence the claim holds.
    \end{proof}
    
    Since there are $3q - l$ vertices in the set $X'$ with a label of $0$, and for each of these vertices, there exists a neighbor in $C_1$ (say $c_i^1)$ that satisfies $f(c_i^1) + f(c_i^2) \geq 2$, hence by this fact and Claim \ref{claim:3}, we can infer that there must be at least $\frac{3q - l}{3}$ vertices among the $c_i^1$s that satisfy $f(c_i^1) + f(c_i^2) \geq 2$. This implies that the weight of the WRD function $f$ is greater than or equal to $\frac{3q - l}{3} + 3q + l$ which simplifies to $4q + \frac{2l}{3}$. However, this conclusion contradicts the initial assumption that $f$ has a weight of at most $4q$. Hence it follows that $l$ must be equal to $0$. Therefore, $f(X') = 0$, and there exist exactly $\frac{3q}{3} = q$ vertices in the set $C_1$ that satisfy the condition $f(c_i^1) + f(c_i^2) = 2$. As a result, $C' = \{c_i \in C~\vert~f(c_i^1) + f(c_i^2) = 2\}$ forms an exact cover of the set $X$.
\end{proof}

Hence, by Theorem \ref{th:2} and the fact that membership of the problem in NP is easily verifiable, it can be concluded that the DECIDE-WRD is NP-complete.

Now, in contrast to the aforementioned result, we show that the MIN-WRD problem is polynomial-time solvable for star convex bipartite graphs with bounded degrees. Now, we state a theorem that is used in proving the above statement.

\begin{thm}\label{th:star}\cite{DBLP:journals/dam/PandeyP19}
    A bipartite graph $G = (X, Y, E)$ is a star convex bipartite graph if and only if there exists a vertex $x$ in $X$ such that every vertex $y$ in $Y$ is either adjacent to $x$ or is a pendant vertex.
\end{thm}

\begin{thm}
    The MIN-WRD problem is polynomial-time solvable for star convex bipartite graph of bounded degree.
\end{thm}

\begin{proof}
    Consider a connected star convex bipartite graph $G=(X\cup Y,E)$, where $\Delta(G)\leq k$. By Theorem \ref{th:star}, there exists a vertex $x$ in $X$, such that every $y\in Y$ is either adjacent to $y$ or is a pendant vertex. Now let $Y_0=N(x)$ and $X_0=N_G(Y_0)$. It is easy to observe that $X_0=X$, because $G$ is connected. Note that $\vert Y\vert=\vert N(x)\vert\leq k$ and $\vert X_0\vert\leq k(k-1)+1$. Hence $\vert V(G)\vert= \vert Y_0\vert + \vert X\vert+\vert Y\setminus Y_0\vert\leq k+k(k-1)+1+k^2(k-1)=k^3+1$. Since the size of $V(G)$ is bounded by $k^3+1$, the MIN-WRD problem can be solved in constant time.
\end{proof}

\section{NP-hardness for Split graphs}
\label{SEC:5}
In this section, we establish the NP-hardness of the problem for split graphs. This result is proved by providing a polynomial-time reduction from the Red Blue Dominating Set problem. Given a bipartite graph $G = (B \cup R, E)$, the Red Blue Dominating Set problem asks to identify a minimum cardinality set $S\subseteq R$ that dominates the set $B$. The decision version of this problem is formulated as follows:
\medskip

\noindent\underline{\textsc{Red Blue Dominating Set decision} problem (DECIDE-RBD)}
\\
[-13pt]
\begin{enumerate}
  \item[] \textbf{Instance}: A bipartite graph $G=(R\cup B,E)$ and a positive integer $k\leq \vert R\vert$.
  \item[] \textbf{Question}: Does there exist $S\subseteq R$ such that $N(S)=B$ and $\vert S\vert \leq k$?
\end{enumerate}

The DECIDE-RBD problem is known to be NP-hard \cite{DBLP:journals/corr/abs-2203-10630}. Now we describe a reduction which maps an instance of the DECIDE-RBD problem to an instance of the DECIDE-WRD problem.

\noindent \textbf{Construction $\mathcal{A}_3$:} Given a bipartite graph $G=(R\cup B,E)$, where the partition consists of $R=\{r_1,\ldots,r_l\}$ and $B=\{b_1,\ldots,b_h\}$. We create the graph $H(V(H),E(H))$ in the following way: $V(H)=R\cup B\cup R'$ where $R'=\{r_1',\ldots,r_l'\}$. $E(H)$ is defined as follows:
\begin{enumerate}
    \item Include all the edges from the original graph $G$, preserving its structure ($E$).
    \item Introduce edges connecting each pair of vertices $r_i$ and $r_j$, where $i\neq j$ and $i, j\in [l]$. This step ensures that all vertices in $R$ are pairwise connected.
    \item Add edges between each vertex $r_i$ and its corresponding counterpart $r_i'$, where $i\in [l]$.
\end{enumerate}
    
Note that in the resulting graph $H$, $R$ is a clique and $R'\cup B$ is an independent set. Hence $H$ is a split graph. Also, the entire construction process can be accomplished in polynomial-time. The reduction can be understood more precisely by the example shown in Figure $\ref{fig:4}$.

\begin{figure}[h!]
    \centering
    \includegraphics[width=0.7\linewidth]{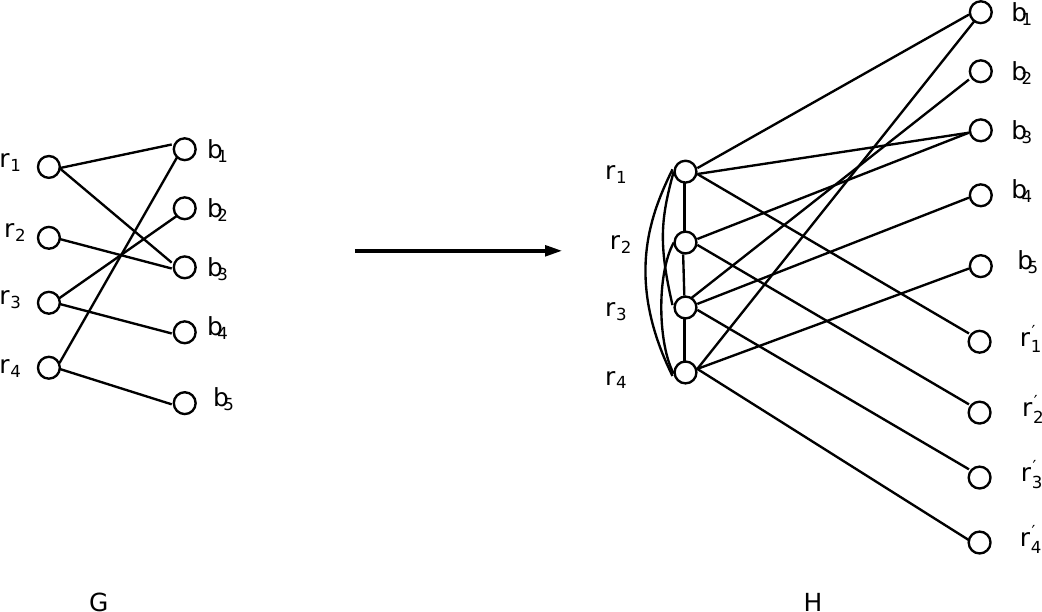}
    \caption{The Construction $\mathcal{A}_3$ applied on the Example 1}
    \label{fig:4}
\end{figure}

\begin{thm}\label{th:3}
    $(G,k)$ is a ``yes" instance of the RBD problem if and only if $(H,k+\vert R\vert)$ is a ``yes" instance of the DECIDE-WRD problem. 
\end{thm}

\begin{proof}
        Given that $(G=(R\cup B,E),k)$ is a ``yes" instance of the DECIDE-RBD problem, implying the existence of a set $S\subseteq R$ with $N(S)=B$ and $\vert S\vert\leq k$, we can construct a WRD function $f:V(H)\rightarrow \{0,1,2\}$ as follows:

        \begin{enumerate}
            \item For each vertex $u$ in the set $R\setminus S$, set $f(u)=1$.
            \item For each vertex $u$ in the set $S$, set $f(u)=2$.
            \item For all other vertices $u$ in $H$, set $f(u)=0$.
        \end{enumerate}
        
        This construction guarantees that $f$ satisfies the properties of a WRD function and $w(f)\leq k+\vert R\vert$. Consequently, we have established that $(H,k+\vert R\vert)$ is a ``yes" instance of  the DECIDE-WRD problem.


        Conversely, let $(H,k+\vert R\vert)$ be a ``yes" instance of the DECIDE-WRD problem, which implies the existence of a WRD function $g$ with $w(g)\leq k+\vert R\vert$. Without loss of generality we assume that $w(g)=\gamma_r(G)$.

        Now, it can be observed that for each vertex $r_i\in R$, $g(r_i)+g(r_i')\geq 1$. Additionally, if we have $g(b_i)\geq 1$ for a vertex $b_i\in B$, then there must exist at least one vertex $r_j\in N_R(b_i)$ such that $g(r_j)+g(r_j')= 1$; as if $g(r_j)+g(r_j')\geq 2$, we may assume $g(r_j)=2$ and $g(r_j')=0$ for every $r_j\in N_R(b_i)$. In that case we can define a WRD function $f$, with $f(b_i)=0$ and $f(v)=g(v)$ for every $v\in V(H)\setminus \{b_i\}$. This implies $w(f)<w(g)$, leading to a contradiction.

        We can leverage this observation to construct another WRD function $g'$ as follows:
        \begin{enumerate}
            \item Pick a vertex $b_i\in B$ with non zero $g'$ label. Set $g'(r_j)=2$ and $g'(r_j')=0$ for a vertex $r_j\in N_R(b_i)$ with $g(r_j)+g(r_j')= 1$.
            \item Set $g'(b_i)=0$ for the vertex $b_i$ associated with the above condition.
            \item For all other vertices $u$ in $H$, set $g'(u)=g(u)$.
        \end{enumerate}
        This construction ensures that $g'$ is a WRD function of same weight as $g$. By repeating this process a finite number of times, we can conclude that there exists a WRD function $f$ with the same weight as $g$, such that $f(B)=0$, and for each vertex $b_i\in B$, there exists a vertex $r_j\in N_R(b_i)$ such that $f(r_j)+f(r_j')= 2$.

        Hence, the set $S=\{r_j\in R~\vert~f(r_j)+f(r_j')=2\}$ dominates the set $B$, that is $N(S)=B$ and $\vert S\vert \leq k$, implying that $(G,k)$ is indeed a ``yes" instance of the DECIDE-RBD problem.
\end{proof}

Hence, by Theorem \ref{th:3}, it can be concluded that the DECIDE-WRD problem is NP-complete for split graphs.

\section{Approximation results}\label{SEC:7}
To the best of our knowledge, there does not exist any approximation result in the literature for the MIN-WRD problem. In this section, we give some approximation results. 
First, we define the minimum dominating set problem.
\\

\noindent\underline{\textsc{Minimum Dominating Set} problem}
\begin{enumerate}
  \item[] \textbf{Instance}: A graph $G=(V,E)$.
  \item[] \textbf{Solution}: A subset $S$ of $V(G)$ with $\vert S\vert =\gamma(G)$.
\end{enumerate}

The following theorems are from previous literature.

\begin{thm}\cite{DBLP:books/daglib/Cormen}\label{th:cormen}
    The MINIMUM DOMINATION SET problem in a graph with maximum degree $\Delta$ can be approximated with an approximation ratio of $1 + ln(\Delta + 1)$.
\end{thm}

In \cite{DBLP:books/daglib/Cormen}, a $1+ln(\Delta+1)$ factor approximation algorithm is given for the MINIMUM DOMINTING SET problem. Let us name this approximation algorithm as \emph{APPROX-DOM}. Next, we show that the MIN-WRD problem can also be approximated within a logarithmic factor. 

\begin{thm}\label{th:approx}
     The MIN-WRD problem in a graph with maximum degree $\Delta$ can be approximated with an approximation ratio of $2(1 + ln(\Delta + 1))$
\end{thm}

\begin{proof}
    Given a graph $G=(V,E)$ with maximum degree $\Delta$, we apply APPROX-DOM on $G$ and get a dominating set $S$. Clearly $\vert S\vert \leq (1+ln(\Delta+1))\gamma(G)$. Now we define $f:V\to \{0,1,2\}$ as: $f(u)=2$ for each $u\in S$ and $f(u)=0$, otherwise. It is easy to observe that $f$ is a weak roman dominating function and $w(f)= 2\vert S\vert\leq 2(1+ln(\Delta+1))\gamma(G)\leq 2(1+ln(\Delta+1))\gamma_r(G)$. Hence we get a polynomial time $2(1+ln(\Delta+1))$-factor approximation algorithm for the MIN-WRD problem.
\end{proof}

Hence by Theorem \ref{th:approx} it can be concluded that the MIN-WRD problem is in APX for bounded degree graphs. Now we show the APX-hardness of the problem. Before showing the APX-hardness we state the definition of an $L$-reduction.

Given two NP optimization problems $F$ and $G$ and a polynomial time transformation $\mathcal{A}$ from instances of $F$ to instances of
$G$, we say that $\mathcal{A}$ is an \emph{$L$-reduction} if there are positive constants $\alpha$ and $\beta$ such that for every instance $x$ of $F$
\begin{enumerate}
    \item $OPT_G(h(x)) \leq \alpha OPT_F (x)$.
    \item For every feasible solution $y$ of $h(x)$ with objective value $m_G(h(x), y) = c_2$ we can find in polynomial time a solution $y'$ of $x$ with $m_F(x, y') = c_1$ such that $\vert OPT_F(x) - c_1\vert \leq \beta\vert OPT_G(h(x)) - c_2\vert$.
\end{enumerate}

\begin{thm}
    The MIN-WRD problem is APX-hard for graphs with maximum degree $4$.
\end{thm}

\begin{proof}
    The MINIMUM DOMINATING SET problem is APX-hard for graphs with maximum degree $3$ \cite{DBLP:journals/tcs/AlimontiK00}. It is enough to construct an $L$-reduction $\mathcal{A}$ from the MINIMUM DOMINATING SET problem to the MIN-WRD problem. Let $G$ be a graph with maximum degree $3$, then we construct $G'$ by adding to each vertex $v\in V(G)$ a path $P_4 (wxyz)$. Clearly, $G'$ is a graph with maximum degree $4$. This reduction in given in \cite{weak_rom_1} to show the NP-completeness of the MIN-WRD problem. We show that this is also an $L$-reduction. 

    Let $D$ be a dominating set of $G$ of cardinality $\gamma(G)$ and $g$ be a WRD function of $G'$ with $w(g)=\gamma_r(G')$. By Lemma $23$ of \cite{weak_rom_1}, $w(g)=\vert D\vert+2\vert V(G)\vert$. Note that $G$ is a graph with maximum degree $3$, hence $\vert V(G)\vert\leq \sum_{v\in D}(deg_G(v)+1)\leq 4\vert D\vert$. So, $w(g)=\vert D\vert+2\vert V(G)\vert\leq 9\vert D\vert$. Hence $\gamma_r(G')\leq 9\gamma(G)$.
    
    Now, Let $f$ be a WRDF of $G'$. Then by the proof of Lemma $23$ in \cite{weak_rom_1}, $f(P_v)\geq 2$, for every $v\in V(G)$. Now three cases may appear.
    
    If $f(P_v)= 3$, and $f(v)=0$ for some $v\in V(G)$, then $f(w)+f(x)+f(y)+f(z)= 3$, then we reconfigure as: $f'(v)=f'(w)=f'(y)=1$, $f(x)=f(z)=0$ and $f'(a)=f(a)$ otherwise. It is easy to observe that $f'$ is a WRDF on $G'$ with $w(f')\leq w(f)$.

    If $f'(P_v)\geq 4$ and $f'(v)=0$ for some $v\in V(G)$, then we reconfigure as: $f''(v)=2$, $f''(w)=f''(y)=1$, $f''(x)=f''(z)=0$ and $f''(u)=f'(u)$ otherwise. It is easy to show that $f''$ is a WRDF on $G'$, with $w(f'')\leq w(f')$.
    
    If $f''(P_v)=2$, then $f''(v)=0$. If $f''(x)=1$, then $f''(w)=0$, which implies that $v$ is adjacent to some $u\in (V_1\cup V_2)\cap V(G)$. If $f''(x)=0$, then clearly $f''(w)=1$ and $f''(v)=0$. But an attack on $x$ can only be defended by moving the legion from $w$ to $x$. This implies $v$ is defended by some $u\in (V_1\cup V_2)\cap V(G)$.

    Now let $S=(V_1\cup V_2)\cap V(G)$, from the above cases, it can be concluded that $S$ is a dominating set of $G$. Then $\vert S\vert \leq w(f'')-2\vert V(G)\vert\leq w(f)-2\vert V(G)\vert$. Hence $\vert S\vert-\vert D\vert\leq w(f)-2\vert V(G)\vert-(w(g)-2\vert V(G)\vert)\leq w(f)-w(g)$. As a result, $\mathcal{A}$ is an $L$-reduction with $\alpha=9$ and $\beta=1$. Hence the theorem follows.
\end{proof}

\section{Conclusion}
\label{SEC:8}
This work presents an extension of the existing algorithmic research conducted on the weak roman domination problem. The NP-completeness of the problem has been proved for specific classes of graphs. 
Also, a polynomial-time algorithm has been proposed for solving the problem in the context of $P_4$-sparse graphs, a well-known extension of the class of cographs. Later we have shown that the problem admits a $2(1+ln(\Delta+1))$ approximation algorithm for graphs with maximum degree $\Delta$. We have also shown that the problem is APX-complete for graphs with maximum degree $4$. It would be of interest to investigate the computational complexity of the problem for other graph classes, such as distance hereditary graphs, permutation graphs, convex bipartite graphs etc. 

 \bibliographystyle{elsarticle-num} 
 \bibliography{Arxiv_version}

\begin{thebibliography}{10}
\expandafter\ifx\csname url\endcsname\relax
  \def\url#1{\texttt{#1}}\fi
\expandafter\ifx\csname urlprefix\endcsname\relax\def\urlprefix{URL }\fi
\expandafter\ifx\csname href\endcsname\relax
  \def\href#1#2{#2} \def\path#1{#1}\fi

\bibitem{COCKAYNE200411}
E.~J. Cockayne, P.~A. Dreyer, S.~M. Hedetniemi, S.~T. Hedetniemi, Roman domination in graphs, Discrete Mathematics 278 (2004) 11--22.

\bibitem{weak_rom_1}
M.~A. Henning, S.~T. Hedetniemi, Defending the roman empire--a new strategy, Discret. Math. 266 (2003) 239--251.

\bibitem{DBLP:journals/dmgt/Henning02}
M.~A. Henning, A characterization of roman trees, Discuss. Math. Graph Theory 22 (2002) 325--334.

\bibitem{stewart1999defend}
I.~Stewart, Defend the roman empire!, Scientific American 281 (1999) 136--138.

\bibitem{liu2010weak}
C.-S. Liu, S.-L. Peng, C.~Y. Tang, Weak roman domination on block graphs, in: Proceedings of The 27th workshop on combinatorial mathematics and computation theory, Providence University, Taichung, Taiwan, 2010, pp. 86--89.

\bibitem{weak_rom_2}
M.~Chapelle, M.~Cochefert, J.~Couturier, D.~Kratsch, R.~Letourneur, M.~Liedloff, A.~Perez, Exact algorithms for weak roman domination, Discret. Appl. Math. 248 (2018) 79--92.

\bibitem{Defn_P4sparse}
G.~Bagan, H.~B. Merouane, M.~Haddad, H.~Kheddouci, On some domination colorings of graphs, Discret. Appl. Math. 230 (2017) 34--50.

\bibitem{GONZALEZ1985293}
T.~F. Gonzalez, Clustering to minimize the maximum intercluster distance, Theor. Comput. Sci. 38 (1985) 293--306.

\bibitem{DBLP:journals/dam/PandeyP19}
A.~Pandey, B.~S. Panda, Domination in some subclasses of bipartite graphs, Discret. Appl. Math. 252 (2019) 51--66.

\bibitem{DBLP:journals/corr/abs-2203-10630}
N.~Abbas, Red domination in perfect elimination bipartite graphs, CoRR abs/2203.10630 (2022).

\bibitem{DBLP:books/daglib/Cormen}
T.~H. Cormen, C.~E. Leiserson, R.~L. Rivest, C.~Stein, Introduction to Algorithms, 3rd Edition, {MIT} Press, 2009.

\bibitem{DBLP:journals/tcs/AlimontiK00}
P.~Alimonti, V.~Kann, Some apx-completeness results for cubic graphs, Theor. Comput. Sci. 237 (2000) 123--134.

\end{thebibliography}





\end{document}